\documentclass[a4paper, 10pt]{article} 

\usepackage[centertags]{amsmath}
\usepackage{amsmath}
\usepackage[dutch,british]{babel}
\usepackage{amsfonts}
\usepackage{amssymb} 
\usepackage{amsthm}
\usepackage{newlfont}
\usepackage{color}
\usepackage{subfig}
\usepackage{bbm}
\usepackage{float}
\usepackage{import}
\usepackage{epsfig}

\usepackage{stmaryrd}
\usepackage{mathrsfs}
\usepackage{pstricks,pst-node}
\usepackage{fancyhdr}
\usepackage{graphicx}
\usepackage{fancybox}
\usepackage{geometry}
 \usepackage{psfrag}
 
 \usepackage{multirow}

  \usepackage{cleveref}

\theoremstyle{plain}
  \newtheorem{theorem}{Theorem}[section]
  \newtheorem{corollary}[theorem]{Corollary}
  
  \newtheorem{lemma}[theorem]{Lemma}
  \newtheorem{remark}[theorem]{Remark}
\theoremstyle{definition}
  \newtheorem{definition}{Definition}[section]
  \newtheorem{assumption}[theorem]{Assumption}

\theoremstyle{remark}

\numberwithin{equation}{section}

 \newcounter{smallarabics}
\newenvironment{arabicenumerate}
{\begin{list}{{\normalfont\textrm{\arabic{smallarabics})}}}
  {\usecounter{smallarabics}\setlength{\itemindent}{0cm}
  \setlength{\leftmargin}{5ex}\setlength{\labelwidth}{4ex}
  \setlength{\topsep}{0.75\parsep}\setlength{\partopsep}{0ex}
   \setlength{\itemsep}{0ex}}}
{\end{list}}

\newcounter{smallroman}

\newcommand{\ben}{\begin{arabicenumerate}}
\newcommand{\een}{\end{arabicenumerate}}

\newtheorem*{theorem*}{Theorem}

\newcommand\otimesal{\mathop{\hbox{\raise 1.6 ex
  \hbox{$\scriptscriptstyle\mathrm{al}$}
\kern -0.92 em \hbox{$\otimes$}}}}
\newcommand\oplusal{\mathop{\hbox{\raise 1.6 ex
  \hbox{$\scriptscriptstyle\mathrm{al}$}
\kern -0.92 em \hbox{$\oplus$}}}}
\newcommand\Gammal{\hbox{\raise 1.7 ex
\hbox{$\scriptscriptstyle\mathrm{al}$}\kern -0.50 em $\Gamma$}}

\let\si=\sigma

   \let\Om=\Omega

\newcommand{\caE}{{\mathcal E}}
\newcommand{\caF}{{\mathcal F}}

\newcommand{\caH}{{\mathcal H}}

\newcommand{\bbR}{{\mathbb R}}

\newcommand{\bbZ}{{\mathbb Z}}

\newcommand{\opunit}{\text{1}\kern-0.22em\text{l}}

\newcommand{\bsn}{{\boldsymbol n}}

\newcommand{\norm}{ \|}
\newcommand{\str}{ |}

\newcommand{\e}{{\mathrm e}}

\newcommand{\beq}{ \begin{equation} }
\newcommand{\eeq}{ \end{equation} }
\newcommand{\bet}{ \begin{theorem} }
\newcommand{\eet}{ \end{theorem} }

\newcommand{\bfk}{\mathbf{k}}

\newcommand{\hfield}{l}

\newcommand{\rra}{\rangle}
\newcommand{\lla}{\langle}

\begin{document}

\title{
Stability of the uniqueness regime for ferromagnetic Glauber dynamics under non-reversible perturbations}
\author{Nick Crawford, Wojciech De Roeck}

%

\maketitle
\vspace{0.5cm}

\begin{abstract}
{We prove a general stability property concerning finite-range, attractive interacting particle systems on $\{-1, 1\}^{\bbZ^d}$.  If the particle system has a unique stationary measure and, in a precise sense, relaxes to this stationary measure at an exponential rate, then any small perturbation of the dynamics also has a unique stationary measure to which it relaxes at an exponential rate.   To apply this result, we study the particular case of  Glauber dynamics for the Ising model.  We show that for any non-zero external field the dynamics converges to its unique invariant measure at an exponential rate.  Previously, this was only known for $\beta<\beta_c$ and $\beta$ sufficiently large.  As a consequence of our stability property, we then conclude that Glauber dynamics for the Ising model is stable to small, nonreversible perturbations in the \textit{entire} uniqueness phase, excluding only the critical point.}
%

\end{abstract}

\section{Introduction}

In this paper, we consider stochastic interacting particle systems -- Markov process on $\Om:= \{-1, 1\}^{\bbZ^d}$ with finite range rates.
Probably the most basic question about such systems concerns their \emph{phase diagrams} with respect to variation of physical parameters like interaction strength, density, etc.: What are their stationary states? Is there a unique stationary state? Is there spontaneous breaking of symmetries and/or long-range correlations?
These issues are reasonably well understood in the context of \emph{reversible dynamics}, i.e.\ when the process under study is reversible w.r.t.\ one or more measures.  Indeed, in natural situations, see e.g.\ \cite{kunsch1984time} it is known that the invariant measures are \emph{Gibbs measures} for a given potential and then the question reduces to classifying all the Gibbs measures, a classical challenge of statistical mechanics.  
In physics, the distinction between reversible and non-reversible dynamics corresponds to the distinction between equilibrium and nonequilibrium dynamics. In the former case, the system is coupled simply to a thermal bath, and in the latter case it is, for example, driven by an external non-gradient field,  or coupled to several thermal baths with non-equal temperatures. The study of non-equilibrium dynamics is a major challenge in statistical physics. To avoid confusion with the usage of the term 'equilibrium' in probability, we will henceforth avoid it, but we stress that our work is inspired by this challenge in physics.

Hence, our aim here is to move beyond reversible dynamics, where it is much harder to formulate general truths.  An example of such an attempt to find a general rule is the \emph{positive rates conjecture} 'In 1d, noisy cellular automata have unique invariant states'. This conjecture has been proved false by a very intricate counterexample \cite{Gacs, Gray2}, even though it is true within the class of attractive dynamics \cite{Gray1}.  Contrasting this,  the restriction of the rates conjecture to reversible dynamics is true: There is 'absence of phase transitions in 1d'. This statement simply means that for short-range interactions, there is a unique Gibbs state on $\bbZ$. As such, the dynamics does not necessarily add much to the question. Not so for nonreversible dynamics! 

Beyond reversibility, one of the conceptually simplest problems  concerns \emph{small perturbations} around reversibility.
For example, if we add  a small reversibility-breaking term, is the phase diagram stable? 
This is the specific question that we address in this note.  Our aims are two-fold.  First, we want to develop a general method allowing us to conclude that stationary measures of small perturbations of a given reversible dynamics are unique, provided that said reversible  dynamics has a unique stationary measure.  Second we wish to \emph{apply} our general theory in a well studied particular case, Ising Glauber dynamics.  Stability of the coexistence phase, that is the stability of the property that there is more than one stationary measure under perturbations, is a more subtle question.  We discuss it a bit further below, but have nothing rigorous to say in this paper.

Ising Glauber dynamics is the natural reversible dynamics associated to a ferromagnetic Ising model with a pair interaction potential (the model and dynamics are described precisely in Section \ref{sec: application}). In this particular case the following picture is known to hold.  The static phase diagram in spatial dimension $d$ may be expressed in terms of two real parameters: the inverse temperature $\beta$ and external field strength $h$.  There is $\beta_c=\beta_c(d)>0$ such that there is a unique Gibbs state for $\beta \leq \beta_c(d)$ or for a nonzero magnetic field $h$. For   $\beta>\beta_c(d)$ at $h=0$, there is coexistence of a $+$ and $-$ phase (magnetic ordering).  A remark worth making here is that in the uniqueness phase it is relatively easy (via monotonicty) to see that the Gibbs measure is also the only stationary measure for the corresponding Glauber dynamics.  That is, there are no non-Gibbs stationary measures.  To our knowledge, the corresponding result is not known in the coexistence region.

Let us give some natural examples of nonreversible perturbations to keep in mind below.  One can imagine making the temperature (which enters as a parameter into the Glauber dynamics) site-dependent, e.g.\ being $\beta \pm \delta$ with $+\delta$ on the even sub-lattice of $\mathbb Z^d$ and $-\delta$ on the odd sublattice. Although our aims are broader, let us note in passing that work related to this model just mentioned appears in both the theoretical physics, \cite{AAM96}, and economics literature, see e.g. \cite{Dur99}. 
In the present paper we prove stability of the uniqueness phase for the Ising model in the following sense (see also \Cref{C:Ising}) \\[0.3mm] 
\begin{theorem}[Stability of uniqueness]
\label{T:Main}
Let the parameters $(\beta,h)$ and the spatial dimension $d$ be such that there is a unique Gibbs state for the Ising model, but excluding the critical point $(\beta_c(d),0)$. Then the weakly perturbed -possibly nonreversible- Glauber dynamics  is still in the uniqueness regime (the required smallness of the perturbation does in general depend on $(\beta,h,d)$).
\end{theorem}
\noindent
A more precise formulation of this result appears below, see in particular \Cref{C:Ising} and the paragraph containing \Cref{T:WSM}.
We regard this as the main result of this paper as it completely settles our question in the uniqueness phase for a touchstone example.

\subsection{Previous results} The literature related to our inquiry is sprawling. At first sight, there are a few papers containing results which seem to be shades of ours, most prominently \cite{holley1985possible}.  
Let us begin our short review by mentioning (very) high temperature techniques. First, one may perturb around independent spin-flips using the independence explicitly in the perturbation expansion, see e.g. \cite{NMW,yarotsky}.  A related technique is a space-time analog of the well-known Dobrushin uniqueness condition from the theory of Gibbs states, see Theorem 4.1 of \cite{LiggettBook} and \cite{D}. Due to their simplicity these methods are robust --  they can both handle arbitrary short ranged classical (Markovian) perturbations.  Unfortunately, they are also very restrictive in that they apply only if there is 'enough' noise in the system; in the one case given by the explicit proximity to independent spin flips and in the other as given by the '$M \leq \epsilon$ - condition'.

In a different direction, recall that in equilibrium statistical mechanics the free energy plays a central role in distinguishing between the uniqueness and phase coexistence regimes.  For translation invariant interactions differentiability of the free energy is equivalent to uniqueness of (translation invariant) Gibbs states. {As a consequence, 'stability' of the uniqueness regime follows when this functional is $C^1$. } To this end, in \cite{Gross}, Gross showed that the Dobrushin uniqueness condition implies the free energy is $C^2$ but that it need not be analytic. Subsequent work by Dobrushin and Shlosman \cite{DS2} found a number of sufficient conditions (complete analyticity) implying analyticity.  Finally, Stroock and Zegarlinski, in \cite{SZ2}, showed that one of these conditions is equivalent to the existence of a log-Sobolev inequality for the associated Glauber dynamics.  Thus, for reversible dynamics, there is, at least at high enough temperature, a circle of ideas which allows one to determine when the uniqueness phase is stable.  Of particular relevance to us, these conditions are manifest both in the properties of the statics and the associated dynamics.

While the developments briefly discussed in the preceding paragraph capture some of the spirit of our inquiry, using log-Sobolev inequalities as basis for perturbing Glauber dynamics in the whole uniqueness phase is not viable.  For one thing, except in dimension two \cite{MOS,SS}, it is not known, even not for the nearest neighbor Ising model,
 if a log-Sobolev inequality holds up to the critical $\beta_c$.  Even worse, if an external field $h \neq0$ is added to the Hamiltonian, there is always a unique Gibbs measure by the Lee-Yang Theorem.   However, for $\beta$ large, it is known that the associated Glauber dynamics cannot satisfy a log-Sobolev inequality.  
 
Of course, complete analyticity and log-Sobolev inequalities are strong sufficient conditions. Depending on one's aim they need not be necessary, especially if we assume that the dynamics under consideration has additional that additional helpful properties. There are two examples of this change of perspective that we want to highlight, both involving the concept of attractivity for interacting particle systems.  (Attractivity is defined formally at \Cref{eq: attractive minus,eq: attractive plus}, while an instance of its important consequence, the existence of monotone couplings of trajectories, is discussed in \Cref{L:Attract}.)  

The first example, \cite{holley1985possible}, proves, similar to our \Cref{thm: main}, that small \textit{attractive} perturbations of attractive particle systems in the uniqueness phase also have unique invariant measures.  Our result is stronger in that it allows the perturbation to be non-attractive.  In addition, our proof is considerably shorter owing to our use of techniques not available at that time. In any case, Holley's result would allow one to show unicity of the invariant measure,  up to $\beta_c$, for perturbations of Ising Glauber dynamics of the type suggested just above  \Cref{T:Main}.

{The second example is of direct relevance to our proofs below.  In \cite{MOS}, Martinelli and Oliveiri show that, given an attractive dynamics, the 'Weak Spatial Mixing' (WSM) condition (see \Cref{WMC} below)
 implies that the infinite-volume dynamics has a unique stationary state $\lla \cdot\rra_{*}$, to which it is strongly exponentially mixing; 
  \beq
 \label{sem}
 \sup_{\sigma_0}\str \lla \si_t(x) \rra_{\sigma_0} - \lla \si(x) \rra_{*} \str \leq Ce^{-c t},
 \eeq
where $\lla \cdot\rra_{\sigma_0}$ is the dynamics started from $\sigma_0$. 
To state the WSM condition, consider the extremal finite volume stationary states $\lla \cdot \rra^{\pm}_{B_L}$ in the cube $B_L$ of sidelength $L$ centered at origin, and corresponding to $\pm$ boundary conditions. Then, the WSM  condition means that
 \beq
 \label{WMC}
 \lla \si(0) \rra^+_{B_L}-\lla \sigma(0) \rra^{-}_{B_L}\leq Ce^{-cL},
 \eeq}
Later, \cite{louis2004ergodicity}  showed that for translation-invariant attractive systems, the WSM condition \eqref{WMC} is in fact equivalent to \eqref{sem}. 

From our point of view, the main advantage of the WSM condition \eqref{WMC} is that we can show it to hold true throughout the uniqueness regime for ferromagnetic interactions.  The WSM condition was established already for $d=2$ in \cite{SS}, for $h=0, \beta< \beta_c, d\geq 2$ in \cite{H} and for  $h\neq 0$ and large $\beta$ in \cite{M}. In the present paper we provide an argument which applies for $h\neq 0$ and arbitrary $\beta>0$. \begin{theorem}
\label{T:WSM}
 For any $h\neq 0$ and finite range ferromagnetic interaction $J_{xy} \geq 0$, the WSM condition \eqref{WMC} holds.
 \end{theorem}

\begin{remark}
As first pointed out to us by D.~Ueltschi \cite{DU}, \Cref{T:WSM} gives a new method for proving uniqueness of Gibbs measures when $h\neq 0$. Traditionally, one relied on the Lee-Yang circle theorem to prove this result (for intermediate values of $\beta$).
\end{remark}

 Martinelli and Olivieri do point out that their framework applies even if the dynamics is {\emph{not} reversible with respect to $\lla \cdot \rra_{*}$}.
At first glance, such a statement goes in the direction we wish to study.  However, for a given non-reversible dynamics, even one close to Glauber dynamics, checking the WSM condition \eqref{WMC} directly seems hard, as one does not expect to have much \textit{a priori} control on the finite volume stationary measures.  

The Martinelli-Olivieri result can be reformulated in terms of  a technique from the theory of Markov chains known as the Propp-Wilson coupling  \cite{propp1998coupling}.  The idea of applying  this coupling to Ising Glauber dynamics first appeared in \cite{SteifvdB}, see Theorem 3.4 there, and was recently employed in \cite{LS},  see Lemma 2.1.  Let us briefly recall the key point of those works (proper definitions appear in \Cref{S:Couple} below).{For any finite range Markov process on $\Om$, let $Y_{t}$ denote the dependence set of the origin tracked backwards from time $t$ to time $0$.}  The authors of \cite{LS} observe that for an attractive dynamics, $\mathbb{E}_{+}[\si_{t}(0)]-\mathbb{E}_{-}[\si_{t}(0)]=2\mathbb{P}(Y_t \neq \varnothing)$, where $\mathbb{E}_{\pm}$ denotes the expectation of the process started from all $+$'s or all $-$'s respectively.   It is easy to see that even if the process is not attractive,  decay of $\mathbb{P}(Y_t \neq \varnothing)$ implies uniqueness of the infinite volume stationary measure, while the rate of its decay controls mixing properties.   This reformulation thus provides us with a useful tool, and our proof is indeed based on it.

{The main restriction on our proof technique  is that the dynamics being perturbed off is attractive and that the WSM condition is satisfied.}  It is not clear to us the extent to which either of these conditions is necessary.  Even for a nearest neighbor ferromagnetic Ising model, one can invent reversible dynamics with the Gibbs measure as the only stationary measure (in the uniqueness regime) and for which our methods do not apply.  For example, consider a dynamics in which nearest neighbor pairs of spins are flipped if and only if they are the same ($++$ or $--$) with rates given by the ratio of the corresponding Gibbs weights.  This is \emph{not} an attractive dynamics and we can therefore say nothing about its stability properties.

To conclude our introductory section, let us say a few words regarding the issue of the stability of coexistence. This seems to be a hard problem which we have not yet made progress on.  The difficulty here is the slow  erosion of droplets of the 'wrong' sign in an extremal low-temperature phase, say '$-$' droplets in a sea of '$+$'.  It is believed that droplets disappear dynamically by a mean curvature flow', which, at a microscopic level, is driven by entropic effects, as opposed to energetic ones. Even for reversible models this picture has not yet been demonstrated, in spite of some recent zero-temperature progress \cite{lacoin2015heat,TonLacoin}.
{The evolution of droplets is more tractable in models where they erode faster, in particular for noisy perturbations of deterministic cellular automata that enjoy the 'eroder property' (erosion of finite droplets in finite time). The best known model in this class is the NEC Toom model \cite{toom1974nonergodic}.  Coexistence of distinct stationary states has been proven for this model and variants of it, see \cite{toom1974nonergodic,de2010phase,fernandez2001non} and \cite{lebowitz1990statistical} for an overview of models.}

%

\subsection*{Acknowledgements}
We thank Christian Maes for introducing us to these beautiful problems. WDR is grateful to the DFG (German Research Fund) and FWO (Flemish Research Fund) for financial support. 

\section{A General Theorem For Attractive Particle Systems}
\subsection{Setup and Result}
\label{S:Gen}
We consider spin systems on the lattice.  As usual, $\Omega=\{-1,1\}^{\bbZ^d}$ is the space of spin configurations $\sigma=(\sigma(x))_{x \in \bbZ^d}$, equipped with the product topology. We consider a Markov dynamics on $\Omega$ defined by local rates for spin updates. Let the space of all continuous functions on $\Om$ be denoted $C(\Omega)$.  Let $c^0_x(\sigma) \geq 0$ denote the rate of flipping $\sigma(x)$ to $-\sigma(x)$ in the configuration $\si$, i.e.\ the generator acting on $C(\Omega)$ is 
\begin{equation} \label{eq: generator}
Lf=\sum_{x} c^0_x(\sigma) \left( f(\sigma^{x})-f(\sigma)\right)
\end{equation}
where $\sigma^{x}(y)= (1-2\delta_{x,y})\sigma(y)$. 
The superscript `$0$' on $c^0$ foreshadows the fact that we will be comparing two dynamics; for the 'unperturbed' dynamics we use the superscript $0$ while the 'perturbed' dynamics will be distinguished by superscripts $1$, as in the perturbed rates $c^1(\si)$.
We always assume these rates to satisfy the following conditions:
\begin{enumerate}
\item Finite range for both $c^0,c^1$: There is a finite $r$ such that, for both $i=0,1$:  $c^i_x(\sigma)=c^i_x(\sigma')$ whenever $\sigma(y)=\sigma'(y)$ for all $\str y-x \str_{\infty}>r$.
\item Uniform bound $\mathop{\sup}\limits_{x,\sigma}c^0_x(\sigma)  <\infty$, for $i=0,1$.
\item  Attractivity for $c^0$: If $\sigma \geq \sigma'$, then 
\begin{align}
\sigma(x)=\sigma'(x)=-1 & \quad \text{implies}\quad  c^0_x(\sigma)  \geq c^0_x(\sigma') \label{eq: attractive minus}\\
\sigma(x)=\sigma'(x)=1 & \quad \text{implies}  \quad 
c^0_x(\sigma')  \geq c^0_x(\sigma)  \label{eq: attractive plus}
\end{align}
\end{enumerate}
The construction of a Feller Markov processes generated by the generator \eqref{eq: generator}, or likewise for $c^0\to c^1$, is standard, see e.g.\ \cite{LiggettBook}. The main case of interest  for us is when the dynamics generated by $c^0$ is reversible with respect to a Gibbs measure.
The above conditions can be easily verified for a given collection of rates. The following assumption, however, is highly nontrivial and can be verified only in specific (reversible) cases. 

Let us denote the law of the Markov process generated by the rates $c^0$ and started from $\sigma$ by $\mathbb{P}^0_\sigma$, then our main assumption reads
\begin{assumption}[Exponentially fast $L^{\infty}$-mixing for $\mathbb{P}^0$] \label{ass: exp strong mixing}
There is a unique invariant state $\mu^0$, such that, for all cylinder functions $f$, 
$$ 
 \sup_{\sigma_0} \left(\mathbb{E}^0_{\sigma_0}(f_t) -\mu^0( f) \right)  \leq C(f)\e^{-tc}
  $$
 where $C(f)$ is translation-invariant: $C(f)=C(\tau_xf)$.
\end{assumption}
As explained in the introduction, the paper \cite{MO} gives a method for checking this assumption under the hypothesis that all finite volume stationary measures satisfy a certain weak mixing condition and that the rates $c^0$ are translation-invariant. 

Now we come to our main technical result, which concerns the perturbed Markov process $\mathbb{P}^1_\sigma$ defined by the rates $c^1$. 
The fact that $c^1$ is a small perturbation of $c^0$ is quantified by a parameter $\epsilon$ defined by 
\beq \label{def: epsilon}
\epsilon \equiv 2 \sup_{x,\sigma} \str c^0_x(\sigma)- c^1_x(\sigma) \str
\eeq
\begin{theorem}\label{thm: main}
If $\epsilon$ is sufficiently small, then the exponential $L^\infty$-mixing  \Cref{ass: exp strong mixing} holds for the process $\mathbb{P}^1$ as well. In particular,   this process has a unique stationary state $\mu^1$. 
\end{theorem}

\subsubsection{Main application: Ising Glauber dynamics} \label{sec: application}

Let us briefly recall the setup for this model.  Let $J: \bbZ^d \times \bbZ^d \to \bbR^+$ be a nonnegative bounded function and write $J_{xy}=J(x,y)$. We require that $J_{xy}=J_{yx}$ and $J_{xy}$ whenever $\str x-y\str_1 > r$ for some finite $r$. Physically, this $J$ plays the role of a ferromagnetic interaction potential (between spins at sites $x,y$) with finite range. We also fix the inverse temperature $\beta >0$ and a magnetic field $h \in \bbR$.   To this data we associate a (formal) Hamiltonian
$$
\caH(\sigma)= -\sum_{x,y} J_{xy}\sigma(x)\sigma(y) -\sum_x h \sigma(x).
$$
and a Glauber dynamics by specifying the rates
$$
c^0_x(\sigma)= \e^{\beta h_{x,\mathrm{eff}}(\sigma) \sigma(x)},\qquad  h_{x,\mathrm{eff}}(\sigma) = h+  \sum_{y} J_{xy}\sigma(y). 
$$
It is clear that these rates satisfy the 3 conditions stated above (finite range, uniform bound and attractivity). The attractivity follows from the \emph{ferromagnetic} nature of the model, i.e.\ from the fact that $J_{xy} \geq 0$.
Let us see when this dynamics has a  unique invariant stationary state.  Let us first check when there is a unique Gibbs state for the Hamiltonian $\caH$, see e.g.\ \cite{israel2015convexity} for precise definitions and background. 
It is well-known that this is the case if $\beta \leq  \beta_c(d)$ for some critical $\beta_c(d)$, or, if $h \neq 0$.   Since the dynamics is attractive, the uniqueness of the stationary state is equivalent to uniqueness of the Gibbs measure, see Theorem 2.16 in Chapter 4 of \cite{LiggettBook}. To verify Assumption \ref{ass: exp strong mixing}  (exponentially fast mixing), we rely on the result of \cite{MO} which deduces this from the Weak Spatial Mixing (WSM) condition \eqref{WMC}. Hence it remains to prove the WSM condition \eqref{WMC} and this can be done in the entire uniqueness regime, except at the critical point $\beta=\beta_c(d),h=0$. 
More concretely, the WSM is proven 
\begin{enumerate}
\item In $d=1$, for all $\beta>0$, by standard transfer matrix methods.
\item In $d=2$,  for $h>0$ or $\beta<\beta_c(2)$, in \cite{SS}
\item  For $h=0, \beta< \beta_c, d\geq 2$ in \cite{H}.
\item  For  $h\neq 0$ and large $\beta$ in \cite{M}.
\item  For  $h\neq 0$ and any $\beta>0$ in the present paper, see \Cref{T1}.
\end{enumerate}

With this statement in hand the following result is  hence immediate from \cite{MO} combined with \Cref{thm: main}.
\begin{corollary}
\label{C:Ising}For parameters $(\beta, h)$ and spatial dimension $d$ such that $h\neq 0$ or $\beta<\beta_c(d)$,  the following statements hold:
\begin{enumerate}
\item The convergence condition \Cref{ass: exp strong mixing} holds for the Glauber dynamics on $\Om$.
\item Uniqueness of the stationary measure is stable to small perturbations in the sense of Theorem \ref{thm: main}.
\end{enumerate}
\end{corollary}

\subsection{Coupling construction and influence clusters}
\label{S:Couple}

Given a pair of processes with respective rates $c^i_x(\sigma)$ satisfying the assumptions set down there, let $\sigma^i_t, i=0,1$ denote the corresponding $\Om$-valued Markov processes. We warn the reader that we shall keep using the notation $\sigma$ for elements of $\Om$.  
We define an overall rate (finite by assumption)
$$
\lambda:=2\sup_{x,\sigma,i} c^i_x(\sigma)
$$

The pair of dynamics can be realized on a single probability space $(\Sigma,\mathbb{P},\caF)$ defined as follows:   We have a collection of independent  rate $\lambda$ Poisson processes $N_x$, indexed by $x\in \mathbb Z^d$.  For each arrival of these Poisson processes, say at $(x,t)$, we associate an independent  uniform random variable $U_{x,t}$ with values in the unit interval $[0,1]$. 
Define the numbers $v^i_x(\sigma) \in [0,1]$ as 
\begin{align}
 v^i_x(\sigma)= (1/\lambda) \begin{cases}  c^i_x(\sigma) & \text{if} \quad \sigma^i(x)=1  \\
 \lambda-  c^i_x(\sigma) &  \text{if} \quad \sigma^i(x)=-1  \end{cases}
\end{align}
Then, at each arrival $(x, t)$ we update $\sigma^i_t(x)$ as
\beq
\sigma^i_t(x) =  \begin{cases} +1  & \text{if} \quad U \geq  v^i_x(\sigma^i_{t-}) \\
-1 &  \text{if} \quad U <  v^i_x(\sigma^i_{t-})  \end{cases}
\eeq
We can check that the law of $\sigma^i$ is indeed given by the $\mathbb{P}^i$. 
Let $\mathcal{F}_{s,t}$ be the sigma field generated by the arrivals  between $s$ and $t$ and their associated $U$'s.
Hence, for any $x$ and $s<t$, the spin $\sigma^i_t(x)$ are measurable w.r.t.\  $\mathcal{F}_{s,t}$ and $\sigma^i_s$.  Fixing the data in $\mathcal{F}_{s,t}$, we  can consider $\sigma^i_t(x) $ almost surely as a measurable function of $\sigma^i_s$.  If we wish to emphasize this dependence, we shall write $\sigma^i_t(x)=\sigma^i_t(x; s, \sigma^i_s) $

\begin{lemma}[Attractivity of update scheme]
\label{L:Attract}
For $s<t$, the function $\sigma^0_s \mapsto \sigma^0_t$ is almost surely increasing, i.e.,
$$
\sigma^0_s \geq \tilde \sigma^0_s  \qquad \text{implies}  \qquad \sigma^0_t \geq \tilde \sigma^0_t
$$
\end{lemma}
\begin{proof}
For concreteness, let us write $\sigma^0_{t-}=(\sigma^0_{t-}(\{x\}^c),\sigma^0_{t-}(x))=:(\eta,\alpha)$ with $\alpha =\pm 1$.
At an arrival $(x,t)$ (and depending on $U=U_{x,t}$), $(\eta,\alpha)$ gets updated to $(\eta,\beta)$ and it suffices to check that such an update is increasing in $(\eta,\alpha)$. The only case that does not follow straightforwardly from Equations (\ref{eq: attractive minus}, \ref{eq: attractive plus}), reduces to the following: take $\alpha=-1, \alpha'=+1,\beta=+1$ and $\eta\leq\eta'$, then $(\eta,\alpha) < (\eta',\alpha')$. The update $(\eta,-1)\to (\eta,1)$ happens when $U \geq   1- (1/\lambda) c^0_x((-1,\eta))$, and the update $(\eta',+1)\to (\eta',+1)$ happens when $U\geq (1/\lambda)c^0_x((+1,\eta'))$.  For the update to be increasing, we hence need that 
$$
(1/\lambda) c^0_x((+1,\eta')) \leq   1-(1/\lambda)c^0_x((-1,\eta))
$$
which follows because of $\lambda \geq 2\sup_{\sigma} c^0_x(\sigma)$.
%
%
\end{proof}

Depending on $U=U_{x,t}$ we call the arrival at $(x,t)$ a 'perturbation arrival' if
$$
(U-v^0_x(\sigma))(U-v^1_x(\sigma)) <0,\qquad \text{for some $\sigma$.}
$$
The idea of this definition is that, if an arrival at $(x,t)$ is \emph{not} a perturbation arrival, then  
$$\sigma^0_{t-}=\sigma^1_{t-} \quad \text{implies}\quad \sigma^0_{t} = \sigma^1_{t}.$$
The probability that a given arrival is a 'perturbation arrival' is bounded by 
$\epsilon$ (as defined in \Cref{def: epsilon}).

\subsection{Influence clusters}
The following description applies both for unperturbed and perturbed dynamics.  As such, we often do not write the superscripts $0/1$ to distinguish between the unperturbed and perturbed dynamics; $Y,W,\ldots$ can stand for $Y^{0/1},W^{0/1},\ldots$.  Let $\str \cdot \str$ refer to the $l^\infty$-norm on $\bbZ^d$ and let
$$
S_x=\{y, \str x-y \str \leq r \}
$$
and recall that both the rates $c^i$ are $r$-local.   Let us say $y,y' \in\bbZ^d $ are $r$-neighbors if $\str y-y' \str \leq r$ and say a subset $A \subset \mathbb Z^d$ is $r$-connected if any pair of vertices $x, y\in A$ can be connected by a path of $r$-neighbors.  Further, we shall say a set $A \subset \bbZ^d \times \bbR$ is $r$-connected if for any pair $(x, s), (y, t)\in A$, there is a piecewise constant-in-time path between them with jumps only at equal time $r$-neighbors.

\begin{definition}[spatial influence sets]
Fixing $s<t$, we say that $y$ influences $\sigma_t(x)$ at time $s$ if there is $\eta \in \Om$ such that $\sigma^i_t(x; s, \eta)\neq \sigma^i_t(x; s, \eta^y)$. Let $Y(s)=Y_{x,t}(s)$ be the (random) set of $y\in \mathbb Z^d$ at time $s$ which influence $\sigma_t(x)$.\end{definition}
Note that the sets $Y(s), s\leq t$ can change only at arrival times $s$.
\begin{definition}[Influence clusters]
For any $x,t$ we call 
$$
W_{x,t} =\overline{\cup_{s\leq t} Y(s)\times s}
$$
the influence cluster.
\end{definition}
Note that $W_{x,t}$ is a $r$-connected set.
Sometimes the influence clusters are too complicated to work with, so we define also the much simpler notion of light-rays.  A lightray $R$ starting at $(x,t)$ is (the graph of) a function $s \mapsto x(s)$  with $s \in (-\infty, t]$ (it is better to think of $s$ running backwards) such that $x(t)=x$ and  $x(s)$ is constant in $s$, except possibly at such $s$ where there is an arrival at $(x(s),s)$, in which case  $x(s_-)=y$ for some $y \in S_{x(s)}$. 

So a lightray is a backwards running path that can jump to $r$-connected sites whenever an arrival hits it.  Note that the definition of lightrays does not involve the variables $U$ whereas the definition of influence clusters involves them in an essential way. 
We need a basic lemma that is merely a restatement of definitions.
\begin{lemma} \label{lem: basic on clusters}
\begin{enumerate}
\item For any $x,t$ and $u \leq t$; 
$$
\big(W_{(x,t)} \cap \{ s\leq u \}\big) \subset  \big(\mathop{\cup}\limits_{y: (y,u) \in W_{(x,t)}}   W_{(y,u)} \big)
$$
\item  For any $x,t$
$$ W_{(x,t)} \subset \big(\mathop{\cup}\limits_{ \text{$R \to (x,t)$}}  \overline R \big),  $$
the union running over lightrays starting at $(x,t)$.
\end{enumerate}
\end{lemma}
 This lemma expresses  that influence clusters can grow at arrivals.   What is not captured by this lemma is the possibility and tendency of influence clusters to die.  This is the basic input from the unperturbed dynamics:
\begin{lemma} \label{lem: decay influence clusters}
 There is a $C$ and $\tau_0$ independent of $x,t$ such that
$$
\mathbb{P}(   Y^0_{x,t}(s) \neq \emptyset ) \leq  C\e^{-(t-s)/\tau_0}
$$
\end{lemma}
\begin{proof}
Let $\sigma^{\pm}_t(x)$ be the value of $\sigma^0_t(x)$ when the (unperturbed) dynamics was started at $t=0$ from all $\pm$. It is a function of the $U$'s in $\caF_{0,t}$ and $\sigma^0_0(y), y \in  Y^0_{x,t}(0)$. By attractivity, we have 
$$
\sigma^{+}_t(x)-\sigma^{-}_t(x) = 2 \chi (Y^0_{x,t}(0)\neq \emptyset)
$$
Taking expectations we get $\mathbb{E}^0_{\text{all $+$}}(\sigma_t(x))-\mathbb{E}^0_{\text{all $-$}}(\sigma_t(x)) =2 \mathbb{P}(Y^0_{x,t}(0) \neq \emptyset) $ and hence the claim follows from  \Cref{ass: exp strong mixing}.
\end{proof}

Now we consider an influence cluster $W^1$ associated to the perturbed dynamics with rates $c^1$. Some of the arrivals in the cluster correspond to 'perturbation arrivals'. However, away from these arrivals, the cluster $W^1$ coincides locally with some cluster $W^0$, by definition.
The following lemma formalizes this.

\begin{lemma}\label{lem: string of clusters}
For the influence cluster $ W^1_{x_0,t_0}$, let  $(x_i,t_i), i=1,2,\ldots$ be the 'perturbation arrivals' in the cluster, i.e.\ $(x_i,t_i) \in W^1_{x_0,t_0}$,  ordered anti-chronologically: $t_{i+1}<t_i$ (this is possible almost surely).  Then  
$$
W^1_{x_0,t_0} \subset \big( W^0_{x_0,t_0} \mathop{\cup}\limits_{i \geq 1} \mathop{\cup}\limits_{y \in S_{x_i}} W^0_{y,t_i}\big).
$$  
\end{lemma}

Since perturbation arrivals are unlikely and since the 'reversible' clusters $W^0$ die out quickly, it is intuitively plausible that also the 'reversible' clusters $W^1$ die out quickly:

\begin{lemma} \label{lem: decay influence clusters noneq}
 There is a $C$ and $\tau_1$ independent of $t,x$ such that
$$
\mathbb{P}(  Y^1_{x,t}(s) \neq \emptyset ) \leq  C \e^{-(t-s)/\tau_1}.
$$
\end{lemma}
The proof is given in the next section, relying on percolation arguments.  Given \Cref{lem:  string of clusters}, the proof of \Cref{thm: main} is immediate since \Cref{lem: decay influence clusters noneq} and the above construction imply that, for a cylinder function $f$ that depends only on $\sigma(x), x \in X$,
$$
\sup_{\sigma,\sigma'}  \str \mathbb{E}^1(f(\sigma_{t}) \str \sigma_s=\sigma) - \mathbb{E}^1(f(\sigma_{t}) \str \sigma_s=\sigma') \str \leq C(f) \e^{-(t-s)/\tau_1},
$$ 
with $C(f)=C \norm f \norm_\infty \str X \str$.

\subsection{Coarse graining: Proof of \Cref{lem: decay influence clusters noneq}}

We partition $\bbZ^d \times \bbR$ with rectangular boxes $B$ as follows. Fix large integers $L,M$ and we define first the box at origin
$$
B_0= \{0,1,\ldots, M-1 \}^d  \times (0,L],
$$
Then, for $\bsn=(k,\hfield) \in \bbZ^{d} \times\bbZ$,
$$B_{\bsn}=B_0+ (Mk,L \hfield)$$

%

The parameter $L$ will have to be chosen large compared to the typical exponential decay time $\tau_0$ that appears in \Cref{lem: decay influence clusters}, i.e.\ we choose $L_t=L \tau_0$ for some large $L \gg 1$.

For each box $B$, we consider also  extended boxes that we denote by $\widetilde B$. They are defined to have the same center as $B$ but with spatial linear size $3$ times bigger, i.e.\ in the above description $ \{0,1,\ldots, M-1 \}^d$ is replaced by $\{-M,\ldots, 2M-1 \}^d$. 
By the (spatial) boundary $\partial\widetilde{B}$ of an extended box, we understand the collection of points $(x,t)$  inside $\widetilde{B}$ such that there is an $(y,t),  y \in S_x $ outside the box.
By the top of the box $B$ we mean the collection of points $(x,t)$  inside the box such that $(x,t+)$ is outside the box.

The choice of $M$ is dictated by the requirement that it is unlikely that a lightray began at the top of $B$ can traverse the corresponding $\tilde{B}$ spatially, i.e. exit $\tilde{B}$ on the spatial boundary. An arrival of the coupled process allows a lightray to move a distance $r$ sideways and the arrivals have rate $\lambda$ clocks. So the probability that a lightray started from a given point at the top of a box $B$ can spatial traverse $\tilde{B}$ is bounded by $\e^{-CL}$, provided we choose 
$$
M-2 \geq c r L,\qquad \text{for some $c>1$}
$$

We call a box $B$ bad if either of the following three events occurs:
\begin{enumerate}
\item A  perturbation update occurs in $\widetilde{B}$.
\item There is a $(x,t) \in B$ such that a lightray $R$ starting at $(x,t)$ reaches $ \partial \widetilde{B}$.
\item  There is an influence cluster $W^0_{x,t}$ with $(x,t)$ in the top of $B$ such that that  $W^0_{x,t} \cap \partial \widetilde{B}=\emptyset$ and $W^0_{x,t}$ exits the box $\widetilde B$ at the bottom, i.e.\ there is a $(y,s) \in W^0_{x,t}$ such that  $s<t$ and $(y,s) \notin \widetilde{B}$ (it necessarily follows that $y$ is in the spatial projection of $\widetilde{B}$).
\end{enumerate}
 The event `box $B$ is bad' is measurable with respect to the data (arrivals and $U$'s) in the extended box $\widetilde{B}$. For 1),2), this is evident, and for 3) it follows from the fact that boxes are open sets at the bottom whereas influence clusters are closed sets. 
 
 The next lemma justifies the intuition that large influence clusters $W^1$ have long connected paths of bad boxes:
\begin{lemma}\label{lem: path of boxes}
If $(y,s) \in W^1_{x,t}$, then there is a path  ${\bsn_1},{\bsn_2},\ldots, {\bsn_m}$ in $\bbZ^{d+1}$ such that 
\begin{enumerate}
\item $ \str \bsn_{i+1}-\bsn_i \str_{\infty}=1$ for $i=1,\ldots,m-1$.
\item The temporal coordinates  $(\hfield_i)_i$ are non-increasing.
\item $B_{\bsn_1}$ contains  $(x,t)$ and $B_{\bsn_m}$ contains $(y,s)$.
\item The boxes $B_{\bsn_i}, 1<i<m$ are bad
\end{enumerate}
\end{lemma} 
\begin{proof}
If an influence cluster $W^0_{x,t}$, with $(x,t)$ in the top of a box $B$, reaches the time $t-L$, then that box $B$ is necessarily bad (either the time $t-L$ is reached at the bottom of $\widetilde{B}$, hence event 3) occurs, or $W^0_{x,t}$ travels far sideways, so as to trigger event 2). Continuing with these considerations, and recalling  \Cref{lem: string of clusters}, the claim follows.
\end{proof}

Hence the problem of large influence clusters is reduced to a percolation problem of bad boxes.  Let us now establish that the probability of being bad is small.
\begin{lemma}\label{lem: bound on bad box}
Fix $L = \lceil N\tau_0 \rceil$ and $M =2r \lceil N\tau_0 \rceil$ for some $N>0$. If $N$ is chosen large enough, and $\epsilon$ small enough (depending on $L$), then 
$$\mathbb{P}(\text{$B_{\bsn}$ is bad}) \leq \e^{-cN},$$
uniformly in $n$. 
\end{lemma} 
\begin{proof}
We just go through the three possibilities in which a box can be bad:\\
1) the probability of a perturbation arrival in a box $\widetilde{B}$ is bounded by $C \epsilon \lambda N^{d+1}$. \\
2) as explained above, we have that the probability of a lightray starting from a given point in $B$ reaching $\partial\widetilde{B}$ is $\e^{-cN}$. Hence, considering all possible starting points (we can reduce to the boundary of $B$) we get   $N^d\e^{-cN}$. \\
3) The probability of $W^0_{x,t}$ reaching the time $t-L_t$ is bounded by $C\e^{-cN}$ by Lemma \Cref{lem: decay influence clusters} and the choice of $L$.  Summing over all $x$ at the top of $B$ we get $N^d \e^{-cN}$.\\
Hence it indeed suffices to choose $L\gg 1$ and $\epsilon \leq (1/\lambda)\e^{-cN}$. 
\end{proof}

We now finish the proof of  \Cref{lem: decay influence clusters noneq}. For any $\bsn \in \bbZ^{d+1}$, let $X_{\bsn}$ be the event that the box $B_{\bsn}$ is not bad. Since $X_{\bsn}$ is measurable with respect to the data of arrivals and $U$'s in $\widetilde B_{\bsn}$, we obtain that $X_\bsn$ is independent of $X_{\bsn'}$, with $\str \bsn' -\bsn\str > C$ for a finite $C$. Then, by \cite{liggett1997domination}, if  $p=\inf_\bsn \mathbb{P}(X_\bsn)\geq 1-C\e^{-cN}$ (see  \Cref{lem: bound on bad box}) is taken close enough to $1$, then the random field $(X_{\bsn} :\bsn \in \bbZ^{d+1})$  dominates stochastically a product random field of $0,1$-valued variables with $\mathbb{P}(1)=\rho$, where $\rho$ can be chosen arbitrarily close to $1$ upon increasing $p$.  Using now standard estimates for product random fields,  \Cref{lem: path of boxes} implies  \Cref{lem: decay influence clusters noneq}.

\section{The Weak Spatial Mixing condition}\label{sec: wsm}

The aim of this section is to prove the WSM condition \eqref{WMC} in the cases where it is not established yet, namely $h\neq 0$ and intermediate inverse temperatures $\beta$, see the discussion in Section \ref{sec: application}. 
Since $\beta$ does not play any role in our argument, from now on we absorb  it into the interaction and field variables: $J'_{xy}=\beta J_{xy}$ and $h'=\beta h$ and we drop the $'$ superscript since we will not need to refer to the original parameters. Since time does not play any role in this section, we write configurations as $\sigma=(\sigma_x)_{x \in \bbZ^d}$.  Let us introduce some additional notation. 
First let  
$$
-\mathcal{H}^{ b, J,h}_S(\sigma):=  \sum_{x\neq y, x, y \in S} J_{xy}\si_x\si_y+\sum_{x\in S} h\si_x + \sum_{x\in S, y \in B^c_L} J_{xy}\si_x b,\qquad    \si\in\{-1, 1\}^{S}
$$
be (minus) the Hamiltonian for an Ising model defined on a set $S \subset B_L$  with interaction $J=(J_{xy})$, external field strength $h$ and boundary condition $b \in \{0,\pm 1\}$ on the exterior $ B^c_L$.  Note, in particular, than there is no boundary interaction corresponding to pairs $\{x, y\}$ with $x\in S$ and $y\in B_L\backslash S$.  Corresponding to this Hamiltonian, we define partition functions and finite-volume Gibbs states in the usual way:
$$
Z_{S}^{ b, J,h} =\sum_{\sigma} \exp{\left(-\mathcal{H}^{ b, J,h}_S(\sigma)\right)},\qquad    \langle F \rangle_{S}^{\pm b, J,h}= \frac{1}{Z_{S}^{ b, J,h}}  \sum_{\sigma} F(\sigma) \exp{\left(-\mathcal{H}^{ b, J,h}_S(\sigma)\right)}
$$ 
where $F$ is any function on $\{-1, 1\}^{S}$.
For the sake of recognizability, we write for the boundary condition $b=f,\pm$ instead of $b=0,\pm 1$ ($f$ stands for free boundary conditions).
As announce, we prove the WSM condition:
\begin{theorem}[Weak Spatial Mixing]
\label{T1}
For any $h\neq 0$ and any bounded finite range interaction $J_{xy}$, there are $c,C>0$ so that
\[
\langle \si_0\rra_{B_L}^{h,+}-\langle \si_0\rra_{B_L}^{h,-}\leq Ce^{-cL}.
\]
\end{theorem}
The rest of Section \ref{sec: wsm} is devoted to the proof. By symmetry, it suffices to consider $h>0$ and we do so henceforth.

\subsection{The variables $\chi,\eta$.}

We introduce a change of variables on pairs of spin configurations which will be instrumental in the proof.
Let the fields $\chi_x, \eta_x$ be defined by
\[
2\chi_x=\si^1_x+\si^2_x, \quad 2\eta_x=\si^1_x-\si^2_x.
\]  
Note that $\chi_x, \eta_x$ are Ising variables subject to the exclusion condition $\chi_x\neq 0$ if and only if $\eta_x=0$.
We have
\beq \label{eq: first product z}
Z^{+, J, h}_{B_L} Z^{-, J, h}_{B_L}=\sideset{}{'}\sum_{\chi, \eta} \exp\left(\sum _{x, y} 2 J_{x, y}[\chi_x \chi_y + \eta_x \eta_y]+ \sum_x 2h \chi_x +\sum_{x\in S, y \in B^c_L} 2J_{xy}\eta_x \right)=:\textrm{I}
\eeq
where $\sideset{}{'}\sum_{\chi, \eta}$ indicates that the exclusion condition is enforced.  
The RHS can further be expressed as
\beq
\label{E:Part}
\textrm{I}=\sum_{V}Z^{f,2J,2h}_{B_L\backslash V}Z^{+, 2J,0}_{V}.
\eeq
On the RHS of this equation, the sum is over subsets $V \subset B_L$ and is obtained from \Cref{eq: first product z} via the identification $V=\{x \in B_L,\chi_x=0\}$. 
In the same way we derive
\beq
\label{E:Diff}
Z^{+,J,h}_{B_L}Z^{-,J,h}_{B_L}[\langle \si_0\rra^{+, J, h}_{B_L}-\langle \si_0\rra^{-, J, h}_{B_L}]={2}\sum_{V: 0 \in V}Z^{f, 2J, 2h}_{B_L\backslash V}Z^{+, 2J, 0}_{ V} \langle \eta_0\rra_{ V}^{+, 2J,0}.
\eeq
Let us note already that $\langle \eta_0\rra_{ V}^{+, 2J,0}\neq 0$ if and only if there is a sequence of vertices $(0=x_0, x_1,\dotsc, x_k)$ such that $x_i \in V$ for all $i$, $J_{x_i, x_{i+1}}\neq 0$ and $x_k\in  B^c_L$.

   \subsection{Basics of random currents}  
Ultimately, we are going to compare the RHS of  \Cref{E:Diff} with 
 \[
 \sum_{V}Z^{f,2J,2h}_{B_L\backslash V}Z^{+, 2J,0}_{V}\:\:(=Z^{+, J, h}_{B_L} Z^{-, J, h}_{B_L}).
 \]
 This comparison requires us to do some surgery on the factor $Z^{+, 2J, 0}_{ V}\langle \eta_0\rra_{ V}^{+, 2J,0}$ in the sum on the RHS of  \Cref{E:Diff}.
Thus for the following discussion, let us fix $V$ such that $0 \in V$ and $0$ is connected to $B^c_L$ by the kernel $J$  $ V$ (cf.\ line below \Cref{E:Diff})
To perform the surgery (see the proof of \Cref{lem: basic bound w}) we shall use the language of random currents \cite{A}, though we don't need the more advanced technology developed in other papers using them (e.g. the switching lemma and its many consequences).
 Let $\mathfrak g$ be an external  'site' to be thought of as representing all sites in $  B^c_L$ (so, unlike what is typically done,  we are not using this site to describe the external field).  We extend the kernel $J_{xy}$ to pairs $x\mathfrak g$ by setting 
 $$J_{x\mathfrak g} = \sum_{y \in B_L^c}  J_{xy}
 $$ 
  For subsets $A\subset V$, let $\mathcal E_A$ denote the set of edges $xy$ with $x, y \in A \cup \{\mathfrak g\}$ such that $x\neq y$ and $J_{xy}\neq 0$.  Given the  $\mathbb{N}\cup\{0\}$-valued vector $\mathbf n=(\mathbf k, \mathbf l)=((k_{xy})_{xy\in \mathcal E_A}, (\hfield_{x})_{x\in A})$ indexed by $\mathcal E_A$ and $A$, let
\beq \label{eq: def w}
W_A^{+, 2J, 2h}(\mathbf n)=\prod_{xy\in \mathcal E_A} \frac{(2J_{xy})^{k_{xy}}}{k_{xy}!} \prod_{x \in A} \frac{(2h)^{\hfield_{x}}}{\hfield_{x}!} 
\eeq
and 
\[\partial \mathbf n=\{x \in V\cup \{\mathfrak g\}: \hfield_x+ \sum_{y\neq x} k_{xy} \text{ is odd}\}.
\]
Here, we note that if $\hfield_x=0$ for $x=\mathfrak g$, the weight $W_A^{+, 2J, 0}(\mathbf n)\neq 0$ only if $l_x=0$ for all $x$.
In this case, we shall write $\mathbf n=\mathbf k=(k_{xy})_{xy \in \mathcal E_A}$.  Similarly, we define $W_A^{f, 2J, 2h}$ by omitting edges ${x\mathfrak g}$ from the product in \eqref{eq: def w}.
By Taylor expansion of exponentials and resummation over $\eta$, we have
\beq
\label{E:RCR2}
Z^{+, 2J,0}_{ V} \langle \eta_0\rra_{ V}^{+, 2J,0}=\sum_{\mathbf k:\partial \mathbf k=\{0, \mathfrak g\}} W^{+, 2J, 0}_V(\mathbf k)
\eeq
Similar formulas are derivable if $h\neq0$, but are not needed here.

Given $\mathbf k=(k_{xy})_{xy\in \mathcal E_A}$, it is convenient to introduce the notation $0\stackrel{\mathbf k}{\leftrightarrow}x$  to indicate that $0$ is connected to  $x$ by edges $uv$ such that $k_{uv}\neq 0$.  In particular $0\stackrel{\mathbf k}{\leftrightarrow}\mathfrak{g}$ for any $\mathbf{k}$ contributing to \eqref{E:RCR2}.  The ($\mathbf k$-dependent) distance $\textrm{dist}_{\mathbf k}(0, x)$ is then the minimal number of edges in a non-zero $\mathbf k$-path from $0$ to $x$.   

\subsection{Clusters around the origin} \label{sec: clusters}
We start from the above formula \eqref{E:RCR2}.
The main idea is to relate this quantity to a partition sum at $h>0$, but this will only work for a subset $T$ (to be defined) of $V$ around the origin. This set is constructed now and the main point (\Cref{lem: bound on f} below) is to establish that it is large enough.  In all that follows, we assume, motivated by \eqref{E:RCR2}, that $0 \in V$ and $0\stackrel{\mathbf k}{\leftrightarrow}\mathfrak{g}$.   Let us define, for any $R>0$, 
\begin{align}
&T_R:=\{x\in  V: \textrm{dist}_{\mathbf k}(0, x)\leq R\}\\[1mm]
&\mathcal F_R:=\{xy\in  \mathcal E_V: x \in T_R, \,  y \notin T_R, \,k_{xy} \neq 0 \} 
\end{align}
Note that $T_R, \mathcal F_R$ depend on $\mathbf k$.  Recalling that the interaction is assumed to be of finite range $r$, we have $T_R\subset B_{rR}$. It will turn out that we only work with $R$ such that $rR < L/2$ and $L$ large, so we can assume that $\mathcal F_R$ contains no edges to $\mathfrak g$.

Let $a$ denote a parameter to be chosen precisely later (but one should think of it as small) and let  
\[
R_0=\inf\{R>L/(4r): |\mathcal F_R|  \leq a |T_R|\}.
\] 
The necessary bound on the size of $T_{R_0}$ is then
\begin{lemma} \label{lem: bound on f}
For any  $a$ small enough, there is $L_0(a)$ such that for $L>L_0(a), R_0< L/(2r)$.  
\end{lemma} 
\begin{proof}
 Due to the finite range condition on the interaction, the number of $J$-edges (that is, $xy$ such that $J_{xy}\neq 0$) incident at any given site is bounded by $C(r) <\infty$. Therefore, 
\begin{equation} \label{eq: perimetry}
\str T_{R+1}  \str -\str T_{R} \str \geq (1/C(r)) \str  \mathcal F_R \str.
\end{equation}

For the sake of a contradiction, let us assume that the condition $ |\mathcal F_R|  \leq a |T_R|$ is violated for all $R$ such that $L/4 <Rr \leq L/3$.  Then the inequality \eqref{eq: perimetry} implies
$$
\str T_{L/(3r)} \str \geq (1+a/C(r))^{L/(12r)}\str T_{L/(4r)} \str.
$$
But this inequality implies exponential growth of $B_{rR}$ in $R$, contradicting $ |B_{R}| \sim R^d$ in $\bbZ^d$.
\end{proof}

Given $a$ small enough, we will henceforth choose $R=R_0$ as defined above (in particular $R_0 < L/(2r)$ by \Cref{lem: bound on f}), and we abbreviate $T=T_{R_0}$ and $\mathcal F= \mathcal F_{R_0}$.

\subsection{Surgery on \Cref{E:RCR2} } \label{sec: surgery}

From \Cref{E:RCR2}, we may use $T=T(\mathbf k)$, defined above, to decompose
\begin{equation}\label{eq: decomp wx}
\sum_{\partial \mathbf k=\{0, \mathfrak g\}}  W^{+, 2J, 0}(\mathbf k)=\sum_X {W(X)}
\end{equation}
where 
\beq \label{def: Wx}
{W(X)}:=\sum_{\partial \mathbf k=\{0, \mathfrak g\}} \mathbf 1\{T=X\} W_{X}^{+, 2J, 0}(\bfk_{\caE_X})W_{V\backslash X}^{+, 2J, 0}(\bfk_{\caE_{V\setminus X}})\prod_{(xy) \in \caE_{\partial_V X}} \frac{J_{xy}^{k_{xy}}}{k_{xy}!}
\eeq
where we have written $\bfk=(\bfk_{\caE_X},\bfk_{\caE_{\partial_V X}},\bfk_{\caE_{V\setminus X}})$
following an $X$-dependent decomposition of the edge set
$$
\caE_V=\caE_X \cup \caE_{V\setminus X} \cup \caE_{\partial}, \qquad  \caE_{\partial_V X}= \{ xy: x \in X,y\in V\setminus X\}.
$$

On the set $X$, we consider random currents $\mathbf m=(\mathbf t, \mathbf \hfield)$ associated to non-zero field, i.e.\ $\mathbf t= t_{xy}, \mathbf \hfield= \hfield_{x} $.
Let $0 \stackrel{\mathbf m}{\Leftrightarrow} X$ denote the condition $ \forall x \in X: 0 \stackrel{\mathbf m}{\leftrightarrow} x$, that is $X$ is connected with respect to the current configuration $\mathbf m$.  Then, naturally,
$$
 T(\mathbf k)=X \qquad  \Rightarrow \qquad  0 \stackrel{\mathbf m}{\Leftrightarrow} X, \qquad \text{if $\mathbf m=(\mathbf t, \mathbf \hfield)$ with $\mathbf t=\bfk_{\caE_X}$} $$  
Our main aim in this section is to connect \eqref{def: Wx} to a sum in which $h>0$. This is achieved by the next lemma.
\begin{lemma}\label{lem: basic bound w}
If $a$ is chosen small enough, there are $C, c>0$ such that $\forall X$,
\beq
\label{E:Key}
W(X)
\leq Ce^{-c L } Z^{+, 2J, 0}_{V\backslash X} \underbrace{\sum_{\partial \mathbf m=\varnothing} \mathbf 1\{0 \stackrel{\mathbf m}{\Leftrightarrow} X \} W_X^{f, 2J, 2h}(\mathbf m)}_{=:K_{X}}.
\eeq
\end{lemma}
\begin{proof}
The main point is to compare $W_{X}^{+, 2J, 0}(\bfk_{\caE_X})$ in \eqref{def: Wx} with a weight in which $h> 0$, namely
\begin{lemma}
For any $\mathbf k$ contributing to \eqref{def: Wx}, 
\beq \label{eq: bound h on w}
W_{X}^{+, 2J, 0}(\bfk_{\caE_X})\leq Ce^{-c \str X\str} G(\bfk_{\caE_X}),\qquad  \text{with $G(\bfk_{\caE_X}) := \sum_{\substack{\mathbf m=(\mathbf t, \mathbf l) \text{s.t.}\\ \mathbf  t=\bfk_{\caE_X}, \partial \mathbf m=\varnothing}} W_X^{f, 2J, 2h}(\mathbf m)$}.
\eeq
\end{lemma}
\begin{proof}
The sum on the RHS is over $\mathbf \hfield$, constrained by $\bfk_{\caE_X}$. It is calculated explicitly as 
\[
\sum_{ \mathbf{\hfield}:  \partial( \bfk_{\caE_X},\mathbf \hfield) =\varnothing} W_X^{f, 2J, 2h}( \bfk_{\caE_X},\mathbf \hfield)= \sinh(2h)^{|X \setminus M\str }\cosh(2h)^{\str M\str }W_{X}^{+, 2J, 0}(\bfk_{\caE_X})
\]
where 
$$M=M(\bfk_{\caE_X})= \{x\in X\backslash \{0\}:   \text{$\sum_{y \in X}k_{xy} $ is even}\}.$$
Hence the claim will follow once we exhibit that (for large enough $L$)  the ratio $\frac{|M|}{ |X\setminus M|} $ can be made arbitrarily large by choosing $a$ small enough.
Recall that $\mathbf k$ arises as a current configuration for which $h=0$ and such that $\partial \mathbf k=\{0,\mathfrak g\}$. Since no site in $X$ can connect directly to $\mathfrak g$, we conclude that any site in $X\setminus M$ (other than $0$) has to have an edge in $\mathcal F$, i.e.\  $ |X\setminus M| \leq |\mathcal F| +1 $. Since, for $\mathbf k$ contributing to \eqref{def: Wx}, $T=X$ and invoking Lemma \ref{lem: bound on f}, we obtain then 
$$
|X\setminus M|  \leq 1 + a |X|
$$
which settles the claim.
\end{proof}

We now end the proof of \Cref{lem: basic bound w}. Plugging \eqref{eq: bound h on w} into \eqref{def: Wx} 
\beq \label{def: Wx2}
{W(X)} \leq C  \sum_{\partial \mathbf k=\{0, \mathfrak g\}}  \e^{-c |X|}  \mathbf 1\{T=X\} \,   G( \bfk_{\caE_X}) \,   W_{V\backslash X}^{+, 2J, 0}(\bfk_{\caE_{V\setminus X}})\prod_{xy \in \caE_{\partial X} } \frac{J_{xy}^{k_{xy}}}{k_{xy}!}
\eeq
Let us first resum $\mathbf{k}_{\caE_{V\setminus X}}$, keeping $\bfk_{\caE_X}, \bfk_{\caE_\partial}$ fixed: 
$$
\sum_{\mathbf{k}_{\caE_{V\setminus X}}:\, \partial\mathbf k=\{0,\mathfrak{g}\}}  W_{V\backslash X}^{+, 2J, 0}(\bfk_{\caE_{V\setminus X}}) =  \left\lla\prod_{y\in P}\si_y \right\rra^{+, 2J, 0}_{V\setminus X}  \leq 1
$$
where $P= \{y \in V\setminus X:  \text{$\sum_{x\in X}k_{yx}$ is odd}\}$.  Next, we sum over $\bfk_{\caE_\partial}$ while keeping $\caF$ fixed.  This leads to a factor $C_J^{|\caF|}$ (from the last factor in \eqref{def: Wx2}).
%
We also note that, given $T=X$, the only constraint on $\bfk_{\caE_X}$ is $0 \stackrel{\mathbf k}{\Leftrightarrow} X$. This leads to 
\beq \label{def: Wx3}
W(X) \leq C\e^{-c |X|}  \left(\sum_{ \bfk_{\caE_X}} \mathbf 1 \{0 \stackrel{\bfk}{\Leftrightarrow} X \}    G( \bfk_{\caE_X})\right)     \sum_{\caF: T=X}   C^{|\caF|}
\eeq
where we have indicated that the last sum over $\caF$ is constrained by $T=X$. 
The quantity between brackets is, by definition, $K_X$ (cf. \Cref{E:Key}). For the sum over $\caF$, we need some combinatorics: By \Cref{lem: bound on f}, we have $|\caF| \leq a |T|$. Furthermore, the number of edges that link to $T$ from $T^c$ is always bounded by $ C(r)|T|$ (with $C(r)$ here being the volume of the sphere with radius $r$). Therefore, the number of choices for $\caF$ reduces to the  
number of ways to pick a subset with up to $a|T|$ elements from $ C(r)|T|$ elements. By standard combinatorics, this is bounded as 
$$
\text{poly}(C(r)|T|)\e^{C(r)|T| f(\frac{a}{C(r)})},\qquad f(p)= -p\log p- (1-p)\log(1-p),\qquad p\in (0,1),
$$
where $\text{poly}(\cdot)$ stands for a polynomial.
Obviously $f(p)\to 0$ as $p\to 0$ and hence, choosing $a$ small enough, we can drop the sum over $\caF$ in \Cref{def: Wx3}, at the expense of readjusting the constant $c$ in $\e^{-c|X|}$. 
Finally, we note that $|T|>L/(4r)$ by the definition of $T$ and hence the lemma follows. 
\end{proof}

\subsection{Proof of \Cref{T1}} 
We are now ready to finish the proof. 
Summarizing the argument given up to this point, in particular \Cref{E:Diff}, \Cref{E:RCR2} and \Cref{lem: basic bound w},  we have shown that
\begin{equation}
\label{E:FinalStretch}
Z^{+,J,h}_{B_L} Z^{-,J,h }_{B_L}[\langle \si_0\rra^{+,J,h}-\langle \si_0\rra^{-,J,h}] \leq  Ce^{-cL} \sum_{V\subset B_L:  0\in V } \sum_{ X \subset V }  Z_{B_L\backslash V}^{f,2J, 2h}  Z^{+,2J, 0}_{V \backslash X} K_X, 
\end{equation}
where $K_X$ appeared on the RHS of \Cref{E:Key}. 
Up to now $K_X$ was defined such that naturally $K_X=0$ when $X=\emptyset$. We extend this function, setting $\tilde{K}_X=K_X$ and setting $\tilde{K}_{\emptyset}=1$.  Trivially, the bound \eqref{E:FinalStretch} gives rise to 
\beq \label{E:FinalStretch2}
Z^{+,J,h }_{B_L} Z^{-,J,h }_{B_L}[\langle \si_0\rra^{+,J,h}-\langle \si_0\rra^{-,J,h}] \leq  Ce^{-cL} \sum_{V\subset B_L} \sum_{ X \subset V }  Z_{B_L\backslash V}^{f,2J, 2h}  Z^{+,2J, 0}_{V \backslash X} \tilde K_X 
\eeq
This is a crucial relaxation in connection with the next lemma (in the case $0\notin A$). 
\begin{lemma}\label{lem: ky}
For any $A \subset B_L$, 
\beq
Z_A^{f, 2J, 2h} = \sum_{Y \subset A} Z_{A\setminus Y}^{f, 2J, 2h} \tilde{K}_Y
\eeq
\end{lemma}
\begin{proof}
We proceed as in the run-up to  \Cref{lem: basic bound w}, in particular we write the analogues of \Cref{eq: decomp wx} and \Cref{def: Wx} but now with $\partial \bfk =\emptyset$. Choosing $Y$ the connected component (by $\bfk$) containing $0$, we get the claim.  
\end{proof}
%
By using first \Cref{E:Part} and then \Cref{lem: ky}, we now derive
\beq
\label{E:FinalStretch two}
Z^{+,J,h }_{B_L} Z^{-,J,h }_{B_L}=   \sum_{U\subset B_L } Z^{f, 2J,2h}_{B_L\backslash U}Z^{+, 2J,0}_{ U} = \sum_{U\subset B_L } \sum_{Y \subset U}  Z^{f, 2J,2h}_{(B_L\backslash U) \backslash Y }Z^{+, 2J,0}_{ U} \tilde{K_Y}.
\eeq 
By the change of variables $Y=X, U=B_L\setminus V $, we see that the RHS of \Cref{E:FinalStretch two} is exactly equal to the double sum in \Cref{E:FinalStretch2}. 
This implies hence $\langle\si_0\rra^{+,J,h}-\langle \si_0\rra^{-,J,h}\leq Ce^{-cL}$, which ends the proof. 
{\flushright \qed}

 \bibliography{noneq}
\bibliographystyle{plain}
 
\end{document}